\newif\ifarxiv
\arxivtrue 

\newif\iftodo
\todotrue

\newif\ifanonymous

\ifarxiv

\documentclass{article}
\usepackage[a4paper, total={6.5in, 10in}]{geometry}
\usepackage[T1]{fontenc}
\usepackage{babel}

\usepackage{authblk}

\usepackage{hyperref}
\hypersetup{
    colorlinks=true,
    linkcolor=blue,
    filecolor=magenta, 
    urlcolor=cyan,
    pdftitle={Asymptotically Optimal Quantum Universal Quickest Change Detection},
    pdfpagemode=FullScreen,
}

\else

\documentclass[conference]{IEEEtran}
\IEEEoverridecommandlockouts

\usepackage{url}

\fi

\usepackage{cite}
\usepackage{amsmath,amssymb,amsfonts,amsthm}
\usepackage{mathtools}
\usepackage{algorithmic}
\usepackage{graphicx}
\usepackage{textcomp}
\usepackage{xcolor}
\usepackage[nameinlink]{cleveref}
\usepackage{aricksMacros}

\usepackage{tikz}
\usetikzlibrary{arrows.meta,positioning,calc,decorations.pathreplacing}
\tikzset{every picture/.style={line width=0.75pt}} 

\newcommand{\RE}[2]{D\left( #1 \| #2 \right)}
\newcommand{\QRE}[2]{S\left(\left. #1 \right\| #2 \right)}

\newcommand{\TFA}{\bar{T}_{FA}}
\DeclareMathOperator{\WADD}{WADD}
\DeclareMathOperator*{\esssup}{ess\,sup}

\newcommand{\deEst}{\hat{p}^w}
\newcommand{\DeEst}{\hat{P}^w}
\newcommand{\bwi}{B_{w,i}}
\newcommand{\mui}{\mu_{w,i}}

\newcommand{\XlM}{X^{\ell, \Mcal}}

\def\BibTeX{{\rm B\kern-.05em{\sc i\kern-.025em b}\kern-.08em
    T\kern-.1667em\lower.7ex\hbox{E}\kern-.125emX}}

\begin{document}

\title{Asymptotically Optimal Quantum Universal Quickest Change Detection\\
}

        



        
    

\ifarxiv

    \author[1]{Arick Grootveld} 
    \affil[1]{Syracuse University}
    \author[1]{Haodong Yang}
    \author[1]{Nandan Sriranga}
    \author[1]{Biao Chen} 
    \author[1]{Venkata Gandikota} 
    \author[1]{Jason Pollack} 

    \date{}

\fi

\maketitle

\begin{abstract}
This paper investigates the quickest change detection of quantum states in a universal setting — specifically, the post-change quantum state is not known {\em a priori}. We establish the asymptotic optimality of a two-stage approach in terms of worst average delay to detection. The first stage employs block POVMs with classical outputs that preserve quantum relative entropy to arbitrary precision. The second stage leverages a recently proposed windowed-CUSUM algorithm that is known to be asymptotically optimal for quickest change detection with an unknown post-change distribution in the classical setting. 
\end{abstract}

\section{Introduction}
\label{section:Introduction}

Change point detection refers to a class of problems whose goal is to detect a change, if it occurs, in the distribution from which a sequence of samples are drawn. For online change point detection, the well-accepted formulation is the so-called quickest change detection (QCD) in which one attempts to minimize the delay to detection subject to a constraint on the false alarm rate, or equivalently, the mean time to false alarm. For the classical setting, Page's test, also known as the cumulative sum (CUSUM) test, was shown to be optimal for QCD \cite{page1954continuous} as it minimizes the worst-case average delay to detection \cite{lorden1971procedures}. 

There have been recent efforts to adapt the QCD to the quantum setting, namely to detect a change in the quantum state given a sequence of quantum systems \cite{akimoto2011discrimination, sentis2016quantum, sentis2017exact, sentis2018online, fanizza2023ultimate, john2025fundamental}. In particular, a quantum CUSUM (QUSUM) algorithm was proposed in \cite{fanizza2023ultimate} that was shown to be asymptotically optimal for quantum QCD for finite-dimensional quantum systems. An extension of this result to infinite dimensions was addressed in \cite{john2025fundamental}. 

Both the classical CUSUM and its quantum counterpart, QUSUM, assume known pre- and post-change distributions or quantum operators. In practice, it is often more realistic to only assume knowledge of pre-change distribution or operator -- the fact that a change occurs is often due to an unknown shift/disruption in distribution/operator in various applications. There have been a number of approaches dealing with QCD with unknown post-change distributions \cite{lai1998information, liang2022quickest,liang2023quickest, xie2023window, liang2024quickest}. Among them, the windowed-CUSUM \cite{liang2024quickest}, which combines a windowed density estimate and the classical CUSUM test, was shown to be optimal -- it achieves asymptotically the same delay to detection as the CUSUM procedure with known post-change distribution. This paper addresses the quantum analog of this problem -- we consider quantum QCD when the post-change quantum operator is unknown. 

There have also been recent related developments in quantum hypothesis testing with unknown operators. This includes the case with a simple null and composite alternative \cite{fujiki2025quantum}, sequential universal quantum hypothesis testing \cite{zecchin2025quantum}, and one-sample and two-sample universal quantum hypothesis testing \cite{grootveld2025towards}. 

In this work, we describe an algorithm for the \textit{quantum universal change point detection problem} in a finite dimensional Hilbert space. Our algorithm combines the windowed-CUSUM result in \cite{liang2024quickest} with optimal measurements for quantum hypothesis testing \cite{hayashi2001asymptotics}, and we show that this algorithm is asymptotically optimal with respect to Lorden's criterion. The general nature of this result could make it useful for various applications in quantum information science, such as entanglement distillation \cite{banerjee2024quantum} and quantum communication \cite{gong2025quantum}.

\section{Background}
\label{section:Background}
Let $\log$ be the natural logarithm. Define the relative entropy between two distributions $P$ and $Q$ as
$$\RE{Q}{P} = \sum\limits_{i=1}^d q_i \log \frac{q_i}{p_i}.$$ 
Similarly, define the quantum relative entropy between two density operators $\sigma$ and $\rho$ to be $$S(\sigma\| \rho) = \Tr[\sigma (\log \sigma - \log \rho)].$$ 
We use $\supp(P)$ to denote the support of distribution $P$, $\dim(\rho)$ for the dimension of the operator $\rho$, and $[\cdot]_+$ to denote the function $[x]_+=\max \{0, x\}$. 

\subsection{Classical Change Point Detection}
\label{subsection:Background_ClassicalCPD}

Let $X_t\sim F_t$ be a discrete-time random process with distribution $F_t$ at time $t$. 
A change point for the sequence is $\nu \geq 0$ 
such that 
\begin{equation}
    F_t = \begin{cases}
        P, & t \leq \nu\\
        Q, & t > \nu,
    \end{cases}
\end{equation}
where $P$ and $Q$ are respectively the pre and post-change distributions in the $d$-dimensional probability simplex,  $\Pcal^d$. We use $\Prob_{\nu}$, $\E_{\nu}$ to denote the probability or expectation of an event conditioned on $\nu$ being the change point.  
Thus, $\nu \rightarrow \infty$ corresponds to no change point occurring, so that the samples are a homogeneous sequence generated by $P$. An algorithm for detecting a change point is characterized by a stopping time $T$, 
the time step at which the algorithm declares a change. The average time to false alarm, $\TFA$, is the expected stopping time when no change has occurred: 
\begin{equation}
    \TFA = \E_{\infty}[T].
\end{equation}
The performance metric most commonly used for QCD is the worst case average detection delay (WADD) \cite{tartakovsky2014sequential}, which characterizes the delay in detecting a change point under the worst case change point location and pre-change samples, i.e.,
\begin{equation}
    \WADD(T) = \sup_{\nu \geq 0} \sup_{\substack{X^{\nu}\\\Prob_{\infty}[X^\nu]>0}} \E_{\nu}[T - \nu | T > \nu, X^{\nu}].
\end{equation}
There is a clear trade-off between $\TFA$ and $\WADD(T)$: while it is desired to have $\TFA$ large and $\WADD(T)$ small, any detection scheme that has a large $\TFA$ is likely to have a longer delay to detection. Lorden established in \cite{lorden1971procedures} that the optimal tradeoff between $\TFA$ and $\WADD$ is given by the following relationship as $\TFA \to \infty$:
\begin{equation}
    \label{equation:LordenConverse}
    \WADD(T) \geq \frac{\log \TFA}{D(Q\|P)} + O(1).
\end{equation}
Lorden went on to show that the so-called CUSUM algorithm achieved this tradeoff asymptotically, in the sense that, 
as $\TFA \to \infty$ 
\begin{equation}
    \label{equation:LordenAchievability}
    \WADD(T_{\text{CUSUM}}) \leq \frac{\log \TFA}{D(Q\|P)} (1 + o(1)).
\end{equation}

The CUSUM algorithm was proposed by Page in \cite{page1954continuous} for the case when both pre- and post-change distributions are known. It computes the CUSUM statistic 
\begin{equation}
    S_k = [S_{k-1} + Z_k]_+,
\end{equation}
where $Z_k = \log \frac{Q[X_k]}{P[X_k]}$ and with initial conditions $Z_0, S_0 = 0$. A change is declared whenever $S_k>h$, where $h$ is the threshold that is designed to control the false alarm rate.  


The CUSUM statistic requires explicit knowledge of $P$ and $Q$ to compute the log-likelihood ratio, and thus the classical CUSUM cannot directly handle situations where the post-change distribution, $Q$, is not specified. There have been many attempts in addressing QCD with unknown post-change distributions. 
In particular, a \emph{non-parametric window-limited adaptive} (NWLA) CUSUM algorithm was proposed in \cite{liang2024quickest}, which we review here. Let $w \in \N$ be the window size. For $n > w$ let $\hat{Q}_n^w$ be an arbitrary method of estimating a probability density using samples $X_{n-w}^{n-1}\triangleq \{X_{n-w},\cdots,X_{n-1}\}$, The empirical log-likelihood ratio of sample $n$ is 
\begin{equation}
    \hat{Z}^w_n = \log \frac{\hat{Q}_{n}^{w}[X_n]}{P[X_n]},
\end{equation}
and for $n \leq w$ take $\hat{Z}^w_n = 0$. Then the NWLA-CUSUM statistic is $\hat{S}^w_n = \left[\hat{S}^w_{n-1} + \hat{Z}^w_n\right]_+$ when $n > w$, and $\hat{S}_n^w = 0$ when $n \leq w$. We state a specific version of Conditions 1 and 2 from \cite{liang2024quickest}:
\begin{condition}[KL-Loss of Estimator]
    \label{condition:KL_Loss}
    For $w$ large enough, there exist constants $0 < C_1 < \infty$ and $\frac{1}{2} < \beta_1 < \infty$ such that 
    \begin{equation}
        \label{equation:Condition1}
        \E_P\left[\RE{P}{\hat{Q}^{w}_n}  \right] \leq \frac{C_1}{w^{\beta_1}}.
    \end{equation}
\end{condition}

Here, we use $E_{P}[\cdot]$ to denote an expectation conditioned on samples following the distribution $P$. 

\begin{condition}[Vanishing Second Moment]
    \label{condition:SecondMoment}
    For $w$ large enough, there exist constants $0 < C_2 < \infty$, and $0 < \beta_2 < 2$ such that 
    \begin{equation}
        \label{equation:Condition2}
        \E_P\left[\left(\log \frac{P[X_{w+1}]}{\hat{Q}_{1}^{w}[X_{w+1}]}\right)^2\right] \leq \frac{C_2}{w^{\beta_2}},
    \end{equation}
\end{condition}
where $\E_P$ denotes the expectation when $X^n \sim P$. 

In \cite[Theorem 2]{liang2024quickest}, the authors showed that if the NWLA-CUSUM algorithm is equipped with $\hat{Q}_n^w$ satisfying \Cref{condition:KL_Loss} and \Cref{condition:SecondMoment}, then the following result holds. 
\begin{theorem}[Theorem 2 of \cite{liang2024quickest}]
    \label{theorem:NWLA_CUSUM_Result}

    Suppose we have a change point problem with a known pre-change distribution $P$ and an unknown post-change distribution $Q$. Assume $\hat{Q}^w_n$ is a kernel density estimator satisfying \Cref{condition:KL_Loss} and \Cref{condition:SecondMoment} with $\beta_1 = \frac{1}{2}$. Then taking $w = (\log\TFA)^{1/2}$ in the NWLA-CUSUM algorithm using this kernel achieves 
    \begin{align}
        \WADD(T_{\text{NWLA}}) \leq \frac{\log(\bar{T}_{FA})}{D(Q \| P)} \left(1 + \Theta\left(   \log(\bar{T}_{FA})^{-1/4}  \right) \right).
    \end{align}
\end{theorem}

\subsection{Quantum Change Point Detection}

Suppose that $\Hcal$ is a $b$-dimensional Hilbert space, and $\Dcal(\Hcal)$ is the set of all density operators on $\Hcal$. In the quantum version of the change point detection problem, we have a quantum source that emits $\xi_\tau \in \Dcal(\Hcal)$. 
Given a pre-change operator $\rho \in \Dcal(\Hcal)$ and post-change operator $\sigma \in \Dcal(\Hcal)$, along with a change point $\nu \geq 0$, the change point problem has the sequence 
\begin{equation}
    \xi_\tau = \begin{cases}
        \rho, & \tau \leq \nu\\
        \sigma, & \tau > \nu
    \end{cases}
\end{equation}

Let $\xi^{\otimes n}$ denote the $n$-fold tensor product of the operator $\xi$. 
For $1\le k\le n$, define  $\xi_{k}^{n} \;:=\; \bigotimes_{i=k}^{n} \xi_i $.  The most general quantum algorithm considered here collects states into blocks of length $\ell$, performs joint measurements on each block using the POVM $\Mcal = \{M_k\}$ inducing classical random variables $\XlM$, and applies classical change point methods to the variables to determine if a change occurred. 
$\XlM_t=k$ indicates that $M_k$ was successful for the $t^{th}$ block. The distribution of the random variables is given by 
\begin{equation}
    \Prob[\XlM_t = k] = \Tr\left[M_k \ \xi_{\ell  \cdot t}^{\ell\cdot (t + 1)}\right]
\end{equation}
The quantum versions of the stopping time $T$, average time to false $\TFA$, and WADD are defined analogously to the classical counterpart defined for $\XlM_t$.  

One of the main results of \cite{fanizza2023ultimate} is a converse that holds for algorithms using the most general methods for recovering information from the states. 
\begin{theorem}[Theorem 2 of \cite{fanizza2023ultimate}]
    \label{theorem:QuantumCPD_Converse}
    Given a change point problem with $S(\sigma \| \rho) < +\infty$. For any detection algorithm with stopping time $T$, $\WADD(T)$ and $\TFA$ 
    must satisfy
    \begin{equation}
        \WADD(T) \geq  \frac{\log \TFA}{S(\sigma\| \rho)}(1-\eps)(1 + o(1)),
    \end{equation}
    for any $\eps > 0$. 
\end{theorem}






In addition,  the Quantum CUSUM (QUSUM) proposed in \cite{fanizza2023ultimate} achieves an asymptotic tradeoff: 
\begin{equation}
    \label{equation:QUSUM_Achievability}
    \WADD(T_{\text{QUSUM}}) \leq \frac{\log \TFA}{\QRE{\sigma}{\rho} (1 - \eps)} + O(1),
\end{equation}
where $\eps$ vanishes as $\TFA \to \infty$. The achievability result is then extended to cases with  post-change states taken from a finite family, and they provide a robustness result for infinite families with suitable discretizations. One of the key pieces in the achievability result is an appropriate choice of projection valued measure (PVM). 



Suppose a measurement on $\sigma$ or $\rho$ reduces the operators to classical probability distributions $P_{\sigma}$ or $P_{\rho}$. The data processing inequality \cite{wilde2013quantum} states that 
\begin{equation}
    \RE{P_{\sigma}}{P_{\rho}} \leq S(\sigma\| \rho).
\end{equation}
A natural question is whether any set of measurements can saturate this inequality. The next theorem gives an affirmative answer in an asymptotic regime. 

To improve the distinguishability of the classical distributions induced from measurement, we collect states into blocks of size $\ell$, and then perform a joint measurement on the whole block. 
\begin{theorem}[Theorem 2 \cite{hayashi2001asymptotics}]
    \label{theorem:MeasurementTheorem}
    Suppose $\rho, \sigma \in \Dcal(\Hcal)$. There exists a PVM $\Mcal_{\ell} = \{M_k\}$ on blocks of size $\ell$, which depends only on $\rho$ (independent of $\sigma$), so that $P^{(\ell)} = (p_{k}^{(\ell)})_{k}$, $Q^{(\ell)} = (q_k^{(\ell)})_{k}$ where
    \begin{align}
        p_k^{(\ell)} = \Tr[M_k \rho^{\otimes \ell}] \quad \text{and,} \quad q_k^{(\ell)} = \Tr[M_k \sigma^{\otimes \ell}]\;
    \end{align}
    such that 
    \begin{equation}
        \label{equation:IneqBetClassicAndQuantRelEnt}
        \lim_{\ell \to \infty} 
        \frac{1}{\ell}\RE{Q^{(\ell)}}{P^{(\ell)}}  = S(\sigma\|\rho)
    \end{equation}
\end{theorem}

Furthermore, \cite[Equation (9)]{hayashi2001asymptotics} shows that we can find such an $\mathcal{M}_{\ell}$ with 
\begin{align}
    \label{align:DimensionOfPVM}
    \abs{\mathcal{M}_{\ell}} &\leq (\ell+1)^{b} 
    \implies \supp(P^{(\ell)}), \supp(Q^{(\ell)}) &\leq (\ell+1)^{b}.
\end{align}


A key observation is that such a PVM can be constructed using only knowledge of $\rho$. This is precisely what is needed in the quantum universal QCD: one wants to convert the quantum QCD into classical QCD while preserving the relative entropy in the absence of any knowledge of the post-change quantum state. 

We assume in the present work that all density operators we deal with are finite dimensional.  Consistent with \cite{fanizza2023ultimate}, we do not consider the case 
with $S(\sigma \| \rho) = \infty$. This corresponds to a trivial tradeoff between the false alarm rate and quickest detection, and have been studied in previous work \cite{sentis2018online}. Without loss of generality we may assume that $\dim(\rho) = \dim(\Hcal) = b$.


\section{NWLA-QUSUM Algorithm for Universal Change Point Detection}
\label{section:Result}
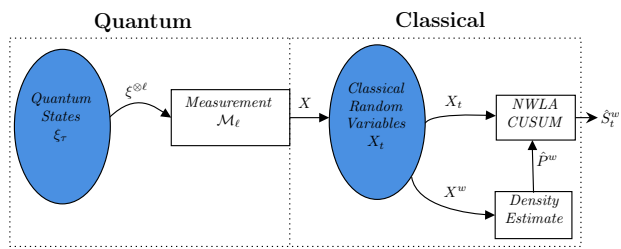
\begin{figure}[ht]
    \centering
    \resizebox{0.5\textwidth}{!}{
        \begin{tikzpicture}[x=0.75pt,y=0.75pt,yscale=-1,xscale=1]

\draw  [fill={rgb, 255:red, 255; green, 255; blue, 255 }  ,fill opacity=1 ] (215,110) -- (320,110) -- (320,160) -- (215,160) -- cycle ;

\draw    (160,135) .. controls (184.31,111.16) and (200.46,125.05) .. (212.54,133.36) ;
\draw [shift={(215,135)}, rotate = 212.42] [fill={rgb, 255:red, 0; green, 0; blue, 0 }  ][line width=0.08]  [draw opacity=0] (8.93,-4.29) -- (0,0) -- (8.93,4.29) -- cycle    ;
\draw    (320,135) -- (352,135) ;
\draw [shift={(355,135)}, rotate = 180] [fill={rgb, 255:red, 0; green, 0; blue, 0 }  ][line width=0.08]  [draw opacity=0] (8.93,-4.29) -- (0,0) -- (8.93,4.29) -- cycle    ;
\draw  [fill={rgb, 255:red, 74; green, 144; blue, 226 }  ,fill opacity=1 ] (397.5,70) .. controls (420.97,70) and (440,100.89) .. (440,139) .. controls (440,177.11) and (420.97,208) .. (397.5,208) .. controls (374.03,208) and (355,177.11) .. (355,139) .. controls (355,100.89) and (374.03,70) .. (397.5,70) -- cycle ;
\draw  [fill={rgb, 255:red, 255; green, 255; blue, 255 }  ,fill opacity=1 ] (503,115) -- (573,115) -- (573,155) -- (503,155) -- cycle ;
\draw    (440,140) .. controls (444.99,131.12) and (484.65,132.43) .. (499.87,133.19) ;
\draw [shift={(502.67,133.33)}, rotate = 183.09] [fill={rgb, 255:red, 0; green, 0; blue, 0 }  ][line width=0.08]  [draw opacity=0] (8.93,-4.29) -- (0,0) -- (8.93,4.29) -- cycle    ;
\draw    (573,135) -- (590,135) ;
\draw [shift={(593,135)}, rotate = 180] [fill={rgb, 255:red, 0; green, 0; blue, 0 }  ][line width=0.08]  [draw opacity=0] (10.72,-5.15) -- (0,0) -- (10.72,5.15) -- (7.12,0) -- cycle    ;
\draw    (428,188) .. controls (432.14,214.42) and (477.01,221.08) .. (499.04,221.92) ;
\draw [shift={(502,222)}, rotate = 180.92] [fill={rgb, 255:red, 0; green, 0; blue, 0 }  ][line width=0.08]  [draw opacity=0] (8.93,-4.29) -- (0,0) -- (8.93,4.29) -- cycle    ;
\draw  [fill={rgb, 255:red, 255; green, 255; blue, 255 }  ,fill opacity=1 ] (502,201) -- (570,201) -- (570,243) -- (502,243) -- cycle ;
\draw    (536,202) -- (535,164.67) -- (535,158) ;
\draw [shift={(535,155)}, rotate = 90] [fill={rgb, 255:red, 0; green, 0; blue, 0 }  ][line width=0.08]  [draw opacity=0] (8.93,-4.29) -- (0,0) -- (8.93,4.29) -- cycle    ;
\draw  [dash pattern={on 0.84pt off 2.51pt}] (72,64) -- (578,64) -- (578,250) -- (72,250) -- cycle ;
\draw  [dash pattern={on 0.84pt off 2.51pt}]  (320,64) -- (320,250) ;
\draw  [fill={rgb, 255:red, 74; green, 144; blue, 226 }  ,fill opacity=1 ] (118.5,74) .. controls (141.97,74) and (161,104.89) .. (161,143) .. controls (161,181.11) and (141.97,212) .. (118.5,212) .. controls (95.03,212) and (76,181.11) .. (76,143) .. controls (76,104.89) and (95.03,74) .. (118.5,74) -- cycle ;

\draw (221,115) node [anchor=north west][inner sep=0.75pt]   [align=left] {\begin{minipage}[lt]{65.08pt}\setlength\topsep{0pt}
\begin{center}
\textit{Measurement}\\$\displaystyle \mathcal{M}_{\ell }$
\end{center}

\end{minipage}};
\draw (86,111) node [anchor=north west][inner sep=0.75pt]   [align=left] {\begin{minipage}[lt]{44.68pt}\setlength\topsep{0pt}
\begin{center}
\textit{Quantum}\\\textit{States}\\$\displaystyle \xi _{\tau}$
\end{center}

\end{minipage}};
\draw (173,102) node [anchor=north west][inner sep=0.75pt]   [align=left] {$\displaystyle \xi ^{\otimes \ell }$};
\draw (326,116) node [anchor=north west][inner sep=0.75pt]   [align=left] {$\displaystyle X$};
\draw (365,102) node [anchor=north west][inner sep=0.75pt]   [align=left] {\begin{minipage}[lt]{44.87pt}\setlength\topsep{0pt}
\begin{center}
\textit{Classical}\\\textit{Random}\\\textit{Variables}\\$\displaystyle X_{t}$
\end{center}

\end{minipage}};
\draw (509,117) node [anchor=north west][inner sep=0.75pt]   [align=left] {\begin{minipage}[lt]{40.12pt}\setlength\topsep{0pt}
\begin{center}
\textit{NWLA}\\\textit{CUSUM}
\end{center}

\end{minipage}};
\draw (455,195) node [anchor=north west][inner sep=0.75pt]   [align=left] {$\displaystyle X^{w}$};
\draw (594,126) node [anchor=north west][inner sep=0.75pt]   [align=left] {$\displaystyle \hat{S}^w_{t}$};
\draw (456,114) node [anchor=north west][inner sep=0.75pt]   [align=left] {$\displaystyle X_{t}$};
\draw (507,202) node [anchor=north west][inner sep=0.75pt]   [align=left] {\begin{minipage}[lt]{42.4pt}\setlength\topsep{0pt}
\begin{center}
\textit{Density }\\\textit{Estimate}
\end{center}

\end{minipage}};
\draw (537,164) node [anchor=north west][inner sep=0.75pt]   [align=left] {$\displaystyle \hat{P}^{w}$};
\draw (412,39) node [anchor=north west][inner sep=0.75pt]  [font=\Large] [align=left] {\textbf{Classical}};
\draw (145,40) node [anchor=north west][inner sep=0.75pt]   [align=left] {\textbf{{\Large Quantum}}};

\end{tikzpicture}
    }
    \caption{Block diagram of NWLA-QUSUM. Measured quantum operators become classical random variables. A window of classical random variables is combined to estimate the density, which is used in the Classical NWLA CUSUM algorithm.}
    \label{fig:AlgorithmDiagram}
\end{figure}

In order to use the result in \Cref{theorem:NWLA_CUSUM_Result} we need a suitable kernel density estimator. Since we are dealing with finite support distributions, we use the following simple kernel over a window of size $w$: Let $\supp(P) = [d]$ be the support of the known distribution $P$. Take $\DeEst$ to be our density estimator, such that $\forall j \in [d]$, and for any $k > w$
\begin{equation}
    \label{equation:KernelDef}
    \deEst_{k,j} = \frac{ 1 + \sum_{i=1}^w \indicator{X_{k-i} = j} }{w+d}.
\end{equation} 
When the time index $k$ is irrelevant, we may omit it. Notice that this density estimator has built-in smoothing: it gives every element in the support of $P$ a weight of at least $\frac{1}{w+d}$. This ensures that the first and second moments of the empirical log-likelihood ratio are finite. 

The achievability results in Eq.\ \eqref{equation:QUSUM_Achievability}, \Cref{theorem:NWLA_CUSUM_Result} are made possible when the block length $\ell$ and the window size $w$ are taken to scale with $\TFA$. However, in order to have our density estimator satisfy \Cref{condition:KL_Loss} and \Cref{condition:SecondMoment}, we will need to carefully balance the rates with which they scale. For example, if the support of the induced classical distribution grows faster than the window size, we should not expect the density estimate to converge to the true density. In \Cref{subsection:Appendix_Condition1} we show that when $d = O(w^{1/2})$, this density estimator satisfies \Cref{condition:KL_Loss} with $\beta_1 = \frac{1}{2}$. In \Cref{subsection:Appendix_Condition2} we show that when $d = O(w^{1/2})$ and $p_i = \Omega(w^{-1/2})$ for all $i \in [d]$, \Cref{condition:SecondMoment} holds for $\beta_2 = \frac{1}{2}$. 


We propose the NWLA-Quantum CUSUM (NWLA-QUSUM) algorithm for change point detection in the universal setting. First measure quantum operators in blocks of length $\ell$ using the measurement scheme from \cite{hayashi2001asymptotics}, giving us an induced classical random variable $X^{\ell, \Mcal}_t$. Then apply the NWLA-CUSUM algorithm to the classical random variables using $\DeEst$ as our kernel. A diagram of the algorithm is found in \Cref{fig:AlgorithmDiagram}. 


\begin{theorem}
    \label{theorem:quantumEmpiricalCUSUM}
    Take $w = (\log\TFA)^{1/2}$, and $\ell = O(\log(\log\TFA))$.
    For the quantum universal setting, the NWLA-QUSUM  algorithm with stopping time $T_{NQ}$, satisfies  
    \begin{equation}
        \WADD(T_{\text{NQ}}) \leq \frac{\log \TFA}{S(\sigma \| \rho) - \eps_{\ell}}\left(1 + O((\log\TFA)^{-1/4}\right),
    \end{equation}

    where $\eps_{\ell} \to 0$ as $\ell \to \infty$. 
\end{theorem}

Combining \Cref{theorem:quantumEmpiricalCUSUM} with \Cref{theorem:QuantumCPD_Converse}, we find that our NWLA-QUSUM algorithm has an optimal asymptotic tradeoff between the $\TFA$ and $\WADD$.

\section{Proofs}
\label{section:Proofs}
\ifarxiv
\subsection{Proof of \texorpdfstring{\Cref{theorem:quantumEmpiricalCUSUM}}{NWLA-QUSUM Result}}
\fi
\label{subsection:Appendix_Theorem4}
\begin{proof}
    Before we can apply the classical result, we need to verify the assumptions imposed in \Cref{condition:KL_Loss} and \Cref{condition:SecondMoment} hold, that $d = O(w^{1/2})$, and $p^{(\ell)}_{\min} = \Omega(w^{-1/2})$. Let $\gamma_{\min}$ be the smallest non-zero element of the spectrum of $\rho$. We take 
    
    \begin{equation}
        \ell = \frac{1}{2} \frac{ \log w }{ \log \frac{1}{\gamma_{\min}} } = O(\log (\log \TFA)).
    \end{equation}

    From \Cref{theorem:MeasurementTheorem} we can find a sequence of PVMs inducing classical distributions which have classical relative entropy approaching the quantum relative entropy as the block size grows. \eqref{align:DimensionOfPVM} shows that the induced distributions have 
    \begin{align*}
        \abs{\supp( P^{(\ell)} )}, \abs{\supp( Q^{(\ell)} )} &\leq (\ell + 1)^b = O((\log w)^b).
    \end{align*}

    Further, the minimum of the spectrum of $\rho^{\otimes \ell}$ is $\gamma_{\min}^{\ell}$: 
    \begin{equation}
        \gamma_{\min}^\ell = \gamma_{\min}^{- \frac{\log w}{2 \log \gamma_{\min} } } = w^{-1/2}.
    \end{equation}
    Since $\rho$ is full rank and $\Mcal_{\ell}$ is a PVM, then $\forall i \in \supp(P^{(\ell)})$, we get
    \begin{align*}
        p_i &= \Tr[\rho^{\otimes \ell} M_i] \geq \gamma_{\min}^{\ell} \Tr[M_i] \geq \gamma_{\min}^{\ell}.
    \end{align*}
    Therefore $p_i = \Omega(w^{-1/2})$.  

    For conciseness we use $X$ to denote $X^{\ell, M_\ell}$ here. Now we proceed with the change point analysis. Similar to the technique used in \cite{fanizza2023ultimate}, we prove the change point result first when $\nu = m \ell$ for $m \in \Z$ (i.e. $\nu$ is an integer multiple of $\ell$ with $m$ being the block index), and relax this assumption later. Under the assumption we have 
    \begin{equation*}
        \WADD(T_{\text{NQ}}) = \sup_{\nu \geq 0} \esssup_{X^m \sim P^{(\ell)}} \E_{\nu}[T_{\text{NQ}} - \nu | T_{\text{NQ}} \geq \nu, X^m].
    \end{equation*}

    From the algorithm, we have $T_{\text{NQ}} = B \ell$ for some $B \in \N$. Therefore, we get 
    \begin{align*}
        &\WADD(T_{\text{NQ}}) \\
        &= \sup_{\nu \geq 0} \sup_{ \substack{X^m \\\Prob_{\infty}[X^m] > 0 } } \E_{\nu}[T_{\text{NQ}} - \nu | T_{\text{NQ}} > \nu, X^m]\\
        &= \ell \cdot \sup_{m \geq 0} \sup_{\substack{ X^m\\ \Prob_{\infty}[X^m]>0 }} \E_{\nu}[B - m | B > m, X^m]\\
        &\overset{(a)}{\leq} \ell \cdot \frac{\log \TFA}{\RE{Q^{(\ell)}}{P^{(\ell)}}}\left(1+\Theta\left((\log\TFA)^{-1/4}\right)\right)\\
        &= \frac{\log \TFA}{\frac{1}{\ell} \RE{Q^{(\ell)}}{P^{(\ell)}}} \left(1 + \Theta\left((\log\TFA)^{-1/4}\right)\right)\\
        &\overset{(b)}{=} \frac{\log \TFA}{ S(\sigma\| \rho) - \eps_{\ell} }\left(1 + \Theta\left((\log\TFA)^{-1/4}\right)\right),
    \end{align*}
    where $(a)$ comes from \Cref{theorem:NWLA_CUSUM_Result}, and $(b)$ comes from \Cref{theorem:MeasurementTheorem}. 
    
    Now, in the more general case where the change can occur inside a block, we can decompose $\nu$ into a block index and a remainder term such that $\nu = m\ell + r$, where $r \in [0,\ell-1]$. Then we have
    { 
    \begin{align*}
        &\WADD(T_{\text{NQ}}) \\
        &\leq \sup_{\nu \geq 0} \sup_{\substack{X^{m+1}\\ \Prob_{\infty}[X^{m+1}]>0 }} \E_{\nu}[T_{\text{NQ}} - \nu | T_{\text{NQ}} > \nu, X^{m+1}]\\
        &\leq \ell + \sup_{\nu \geq 0} \sup_{ \substack{X^{m+1} \\ \Prob_{\infty}[X^m+1]>0} }\E_{\nu}[T_{\text{NQ}} - \nu - r | T_{\text{NQ}} > \nu + r, X^{m+1}]\\
        &\leq \ell + \frac{\log \TFA}{ S(\sigma\| \rho) - \eps_{\ell} }\left(1 + \Theta\left((\log\TFA)^{-1/4}\right)\right).
    \end{align*}
    }

    And since $\ell = O(\log(\log \TFA))$, we get finally
    \begin{equation*}
        \WADD(T) \leq \frac{\log \TFA}{ S(\sigma\| \rho) - \eps_{\ell} }\left(1 + \Theta\left((\log\TFA)^{-1/4}\right)\right).
    \end{equation*}
    
\end{proof}
\vspace{-3em}
\subsection{Condition 1}
\label{subsection:Appendix_Condition1}
{
Here we prove that the kernel described in \eqref{equation:KernelDef} satisfies \Cref{condition:KL_Loss}, 
for $\beta_1 = \frac{1}{2}$ if $d = O(w^{1/2})$.

We start by estimating $\E_P[\log \deEst_i]$ for $i \in [d]$, and use this to get a bound on the KL-Loss in the required form. Let $\bwi = \sum_{k=1}^w \indicator{X_k = i}$, and $\mui = w p_i = \E_P[\bwi]$, since $\bwi$ is a binomial distributed random variable
\begin{align*}
    \E_P[\log\deEst_i] &= - \log(d+w) + \E_P\left[\log\left(1 + \bwi\right)\right].
\end{align*}

By Taylor expanding $\log(1 + \bwi)$ around $1 + \mui$, we get 
{ \small
\begin{align*}
    &\log(1 + \bwi) \\
    &= \log(1 + \mui) + \frac{\bwi - \mui}{1 + \mui} - \frac{(\bwi - \mui)^2}{2(1 + \mui)^2} + O\left(\mui^{-2}\right)\\
    &\implies \E_P[\log(1+\bwi)] \\
    &= \log(1+w p_i) - \frac{w p_i(1-p_i)}{2 (1 + wp_i)^2} + O(w^{-2}),
\end{align*}
}
where the first equality comes from the fact that the 3rd central moment of a binomial distribution is $\E[(\bwi - \mui)^3] = w p_i(1 -p_i)(1 - 2p_i)$. 

Therefore
{\small
\begin{align}
    &\E_{P}[\log \deEst_i] \\
    &= \log\left(\frac{1 + w p_i}{d+w}\right) - \frac{w p_i (1 - p_i)}{2 (1 + wp_i)^2} + O\left(w^{-2}\right),
\end{align}
}

which gives us 
{\small
\begin{align*}
    &\E_P[D(P \| \DeEst)] \\
    &= \sum_{i=1}^d p_i \log\left(\frac{p_i(d+w)}{1 + w p_i}\right) + \frac{w p_i^2 (1 - p_i)}{2(1 +w p_i)^2} - O(w^{-2}).
\end{align*}
}

We bound the first term using $\log(1+x) \leq x$:
{ \small
\begin{align*}
    p_i \log\left(\frac{p_i(d+w)}{1 + w p_i}\right) 
    &= p_i \log \left(1 + \frac{p_i d  - 1}{1 + w p_i}\right)\\
    &\leq p_i \left(\frac{p_i d - 1}{1 + w p_i}\right) \leq \frac{d p_i}{w}.
\end{align*}
}

Now we bound $\frac{wp_i^2(1-p_i)}{2(1+w p_i)^2}$ as follows:
{\small
\begin{align*}
    \frac{w p_i^2 (1 - p_i)}{2 (1 + w p_i)^2} 
    &\leq \frac{w p_i^2}{2 w^2 p_i^2} = \frac{1}{2w},
\end{align*}
}
so that we have 
{ \small
\begin{align*}
    \E_P[D(P\| \DeEst)] &\leq \frac{3}{2}\cdot \frac{d}{w} - O(w^{-2}).
\end{align*}
}

From the assumption $d = O(w^{1/2})$, we get $\frac{d}{w} = O(w^{-1/2})$. Therefore,  \Cref{condition:KL_Loss} is satisfied with $\beta_1 = \frac{1}{2}$, and $C_1$ a suitably large constant. 
}

\subsection{Condition 2}
\label{subsection:Appendix_Condition2}

{ 

Here we show that the second moment  of the log-likelihood ratio for our estimator is bounded per \Cref{equation:Condition2}, in roughly the same way that we showed \Cref{equation:Condition1}. That is, we write $\deEst$ in terms of $\bwi$, approximate the expectation with a second order Taylor series, and then bound the remaining terms. Here, we have the assumptions $d = O(w^{1/2})$ and $p_i = \Omega(w^{-1/2})$.



We Taylor expand $\left(\log\frac{1 + \bwi}{(d+w) p_i}\right)^2$ around $\frac{1 + \mui}{(d+w)p_i}$, giving us
{ 
\begin{align*}
    &\left(\log\frac{1 +  \bwi}{(d+w) p_i}\right)^2 \\
    &= \left(\log\frac{1 + \mui}{(d+w)p_i}\right)^2 + \frac{2 \log\left(\frac{1 + \mui}{(d+w)p_i}\right)}{1 + \mui}(\bwi - \mui) \\
    &\ \ + \frac{(1 - \log\left(\frac{1+ \mui}{(d+w)p_i}\right)}{(1 + \mui)^2}(\bwi - \mui)^2 + O(\mui^{-2}),
\end{align*}
}
where $\mui = w p_i = O(w^{1/2})$. As $\mui$ is a deterministic quantity, applying this Taylor series allows us to evaluate an individual term in the expectation:
{ 
\begin{align*}
    &\E_P\left[\left(\log\frac{\deEst_i}{p_i}\right)^2\right] \\
    &= \left(\log\frac{1 + w p_i}{(d+w)p_i}\right)^2 \\
    &+ \frac{1 - \log\left(\frac{1 + w p_i}{(d+w)p_i}\right)}{(1 + w p_i)^2} \cdot w p_i (1-p_i) + O(w^{-1}).
\end{align*}
}

We can bound $\left(\log\frac{1 + w p_i}{(d+w)p_i}\right)^2$, in the following manner: 
{ 
\begin{align*}
    \abs{\log \frac{1 + wp_i}{(d+w) p_i}} 
    &\leq \log\left(1 + \frac{1}{w p_i}\right) + \log\left(1 + \frac{d}{w}\right)\\
    &\leq \frac{1}{w p_i} + \frac{d}{w}\\
    \implies \left(\log \frac{1 + w p_i}{(d+w)p_i}\right)^2 &\leq \left(\frac{1+ dp_i}{w p_i}\right)^2 = O(w^{-1}).
\end{align*}
}

Of note, if $d = O(w^{1/2})$, then, as $w \to \infty$,
{ 
\begin{equation*}
    \log\left(\frac{1 + w p_i}{d p_i + w p_i}\right) \to 0.
\end{equation*}
}

Therefore, for $\forall \eps > 0$, $\exists N_{\eps} > 0$ such that, $\forall w \geq N_{\eps}$,
{
\begin{equation*}
    1 - \log\left(\frac{1 + w p_i}{d p_i + w p_i}\right) \leq 1 + \eps.
\end{equation*}
}

As a result, for $w \geq N_{\eps}$ we have 
{ 
\begin{align*}
    \E_P\left[\left(\log\frac{\deEst_i}{p_i}\right)^2\right] 
    &\leq \frac{(1 + \eps)}{w p_i} + O(w^{-1}).
\end{align*}
}

Finally, if $w \geq N_{\eps}$, we have 
{
\begin{align*}
    &\E_P\left[\left(\log \frac{\DeEst[X]}{P[X]}\right)^2\right] \\
    \\&= \sum_{i = 1}^d P[X = i] \ \E_p\left[\left.\left(\log \frac{\DeEst[X]}{P[X]}\right)^2 \right| X = i\right]\\
    &\leq \sum_{i=1}^d p_i \frac{1 + \eps}{w p_i} + O(w^{-1}) = \frac{d}{w}(1 + \eps) + O(w^{-1}).
\end{align*}
}

Therefore, from our initial condition, we get $\frac{d}{w} = O(w^{-1/2})$. Thus \Cref{condition:SecondMoment} is satisfied with $\beta_2 = \frac{1}{2}$ and $C_2$ some large constant.  

}

\section{Conclusions}
\label{section:Conclusions}

This paper proposed a two-stage approach for quantum QCD with an unknown post-change quantum operator: the first stage employs PVMs that preserve the quantum relative entropy, while the second stage applies the classical windowed-CUSUM. We establish that such a scheme is asymptotically optimal in minimizing the delay to detection under a false alarm rate constraint. 

An interesting variation of the problem studied here is to relax the tensor product assumption, and instead take $\xi_t = U_t \rho U_t^\dagger$ or $\xi_t = U_t \sigma U_t^\dagger$ where $U_t$ is a unitary operator describing some form of time evolution. Time evolution is a fundamental property of quantum mechanics, and any algorithm with practical ambitions should account for it. 

\ifarxiv
\bibliographystyle{IEEEtran}
\bibliography{refs.bib}

@article{fanizza2023ultimate,
  title={Ultimate limits for quickest quantum change-point detection},
  author={Fanizza, Marco and Hirche, Christoph and Calsamiglia, John},
  journal={Physical Review Letters},
  volume={131},
  number={2},
  pages={020602},
  year={2023},
  publisher={APS}
}

@article{liang2024quickest,
  title={Quickest change detection with post-change density estimation},
  author={Liang, Yuchen and Veeravalli, Venugopal V},
  journal={IEEE Transactions on Information Theory},
  volume={70},
  number={11},
  pages={8072--8086},
  year={2024},
  publisher={IEEE}
}

@article{zecchin2025quantum,
  title={Quantum Sequential Universal Hypothesis Testing},
  author={Zecchin, Matteo and Simeone, Osvaldo and Ramdas, Aaditya},
  journal={arXiv preprint arXiv:2508.21594},
  year={2025}
}

@article{page1954continuous,
  title={Continuous inspection schemes},
  author={Page, Ewan S},
  journal={Biometrika},
  volume={41},
  number={1/2},
  pages={100--115},
  year={1954},
  publisher={JSTOR}
}

@article{lorden1971procedures,
  title={Procedures for reacting to a change in distribution},
  author={Lorden, Gary},
  journal={The annals of mathematical statistics},
  pages={1897--1908},
  year={1971},
  publisher={JSTOR}
}

@article{grootveld2025towards,
  title={Towards Quantum Universal Hypothesis Testing},
  author={Grootveld, Arick and Yang, Haodong and Chen, Biao and Gandikota, Venkata and Pollack, Jason},
  journal={arXiv preprint arXiv:2504.16299},
  year={2025}
}

@book{tartakovsky2014sequential,
  title={Sequential analysis: Hypothesis testing and changepoint detection},
  author={Tartakovsky, Alexander and Nikiforov, Igor and Basseville, Michele},
  year={2014},
  publisher={CRC press}
}

@article{akimoto2011discrimination,
  title={Discrimination of the change point in a quantum setting},
  author={Akimoto, Daiki and Hayashi, Masahito},
  journal={Physical Review A—Atomic, Molecular, and Optical Physics},
  volume={83},
  number={5},
  pages={052328},
  year={2011},
  publisher={APS}
}

@article{sentis2016quantum,
  title={Quantum change point},
  author={Sent{\'\i}s, Gael and Bagan, Emilio and Calsamiglia, John and Chiribella, Giulio and Munoz-Tapia, Ramon},
  journal={arXiv preprint arXiv:1605.01916},
  year={2016}
}

@article{sentis2017exact,
  title={Exact identification of a quantum change point},
  author={Sent{\'\i}s, Gael and Calsamiglia, John and Munoz-Tapia, Ramon},
  journal={arXiv preprint arXiv:1707.07769},
  year={2017}
}

@article{sentis2018online,
  title={Online strategies for exactly identifying a quantum change point},
  author={Sent{\'\i}s, Gael and Mart{\'\i}nez-Vargas, Esteban and Munoz-Tapia, Ramon},
  journal={Physical Review A},
  volume={98},
  number={5},
  pages={052305},
  year={2018},
  publisher={APS}
}

@article{hayashi2001asymptotics,
doi = {10.1088/0305-4470/34/16/309},
url = {https://doi.org/10.1088/0305-4470/34/16/309},
year = {2001},
month = {apr},
publisher = {},
volume = {34},
number = {16},
pages = {3413},
author = {Masahito Hayashi},
title = {Asymptotics of
quantum relative entropy from a representation theoretical
viewpoint},
journal = {Journal of Physics A: Mathematical and General}
}

@article{john2025fundamental,
  title={Fundamental Limits Of Quickest Change-point Detection With Continuous-Variable Quantum States},
  author={John, Tiju Cherian and Gagatsos, Christos N and Bash, Boulat A},
  journal={arXiv preprint arXiv:2504.16259},
  year={2025}
}

@article{xie2023window,
  title={Window-limited CUSUM for sequential change detection},
  author={Xie, Liyan and Moustakides, George V and Xie, Yao},
  journal={IEEE Transactions on Information Theory},
  volume={69},
  number={9},
  pages={5990--6005},
  year={2023},
  publisher={IEEE}
}

@article{lai1998information,
  title={Information bounds and quick detection of parameter changes in stochastic systems},
  author={Lai, Tze Leung},
  journal={IEEE Transactions on Information theory},
  volume={44},
  number={7},
  pages={2917--2929},
  year={1998},
  publisher={IEEE}
}

@article{liang2022quickest,
  title={Quickest change detection with non-stationary post-change observations},
  author={Liang, Yuchen and Tartakovsky, Alexander G and Veeravalli, Venugopal V},
  journal={IEEE Transactions on Information Theory},
  volume={69},
  number={5},
  pages={3400--3414},
  year={2022},
  publisher={IEEE}
}

@inproceedings{liang2023quickest,
  title={Quickest change detection with leave-one-out density estimation},
  author={Liang, Yuchen and Veeravalli, Venugopal V},
  booktitle={ICASSP 2023-2023 IEEE International Conference on Acoustics, Speech and Signal Processing (ICASSP)},
  pages={1--5},
  year={2023},
  organization={IEEE}
}

@article{fujiki2025quantum,
  title={Quantum hypothesis testing for composite alternative hypotheses},
  author={Fujiki, Daichi and Tanaka, Fuyuhiko and Sakashita, Tatsuya},
  journal={Physical Review A},
  volume={111},
  number={3},
  pages={032405},
  year={2025},
  publisher={APS}
}

@book{wilde2013quantum,
  title={Quantum information theory},
  author={Wilde, Mark M},
  year={2013},
  publisher={Cambridge University Press}
}

@article{banerjee2024quantum,
  title={Quantum change point and entanglement distillation},
  author={Banerjee, Abhishek and Bej, Pratapaditya and Bandyopadhyay, Somshubhro},
  journal={Physical Review A},
  volume={109},
  number={4},
  pages={042407},
  year={2024},
  publisher={APS}
}

@article{gong2025quantum,
  title={Quantum-enhanced change detection and joint communication detection},
  author={Gong, Zihao and Guha, Saikat},
  journal={Physical Review A},
  volume={112},
  number={3},
  pages={032604},
  year={2025},
  publisher={APS}
}
\fi

\end{document}